\documentclass[letterpaper,twocolumn,10pt,dvipsnames]{article}
\usepackage{usenix}

\usepackage{amsmath}
\usepackage{algorithm}
\usepackage{algpseudocode}
\usepackage{cleveref}
\usepackage{annotate-equations}
\usepackage{circledsteps}
\usepackage{xcolor}
\usepackage[most]{tcolorbox}
\usepackage{tabularx}
\usepackage{booktabs}
\usepackage{array}

\usepackage{amsthm}
\newtheorem{lemma}{Lemma}

\definecolor{ObservationBlue}{HTML}{006699}
\newtcolorbox{observationboxnonumber}[1][]{%
  enhanced,
  breakable,
  colback=gray!3,
  colframe=black!20,
  coltitle=black,
  fonttitle=\bfseries,
  title={Observation},
  boxrule=0.5pt,
  leftrule=0pt,
  rightrule=0pt,
  toprule=0pt,
  bottomrule=0pt,
  borderline west={2pt}{0pt}{ObservationBlue},
  boxsep=4pt,
  left=7pt,
  right=7pt,
  top=5pt,
  bottom=5pt,
  #1
}

\newtcolorbox[auto counter]{observationbox}[2][]{%
  enhanced,
  breakable,
  colback=gray!3,
  colframe=black!20,
  coltitle=black,
  colbacktitle=gray!10,
  coltitle=MidnightBlue!50,
  fonttitle=\bfseries,
  title={Observation~\thetcbcounter: #2},
  boxrule=0.5pt,
  leftrule=0pt,
  rightrule=0pt,
  toprule=0pt,
  bottomrule=0pt,
  borderline west={2pt}{0pt}{ObservationBlue},
  boxsep=4pt,
  left=7pt,
  right=7pt,
  top=5pt,
  bottom=5pt,
  #1
}

\newcolumntype{L}{>{\raggedright\arraybackslash}X}

\newtcolorbox{apispecbox}[2][]{%
  enhanced,
  breakable,
  colback=gray!3,
  colframe=black!20,
  coltitle=black,
  fonttitle=\bfseries,
  title={API: #2},
  boxrule=0.5pt,
  leftrule=0pt,
  rightrule=0pt,
  toprule=0pt,
  bottomrule=0pt,
  borderline west={2pt}{0pt}{ObservationBlue},
  boxsep=4pt,
  left=7pt,
  right=7pt,
  top=6pt,
  bottom=6pt,
  #1
}

\definecolor{SpecPink}{HTML}{FF99FF}
\definecolor{DecodeOrange}{HTML}{FF6633}
\definecolor{DecodeOrange}{HTML}{DD6633}
\definecolor{PrefillBlue}{HTML}{006699}

\begin{document}

\date{}

\title{\Large \bf Optimizing Agentic Language Model Inference via Speculative Tool Calls} %

\author{
{\rm Anonymous Authors}\\
} %

\author{
    {\rm Daniel Nichols}\\
    Lawrence Livermore National Laboratory
    \and
    {\rm Prajwal Singhania}\\
    University of Maryland
    \and
    {\rm Charles Jekel}\\
    Lawrence Livermore National Laboratory
    \and
    {\rm Abhinav Bhatele}\\
    University of Maryland
    \and
    {\rm Harshitha Menon}\\
    Lawrence Livermore National Laboratory
}

\maketitle

\begin{abstract}
Language models (LMs) are becoming increasingly dependent on external tools.
LM-based agentic frameworks frequently interact with their environment via such
tools to search files, run code, call APIs, etc. Further, modern
reasoning-based LMs use tools such as web search and Python code execution to
enhance their reasoning capabilities. While tools greatly improve the
capabilities of LMs, they also introduce performance bottlenecks during the
inference process. 
In this paper, we introduce novel systems optimizations to address such
performance bottlenecks by speculating tool calls and forcing sequences to
remain resident in the inference engine to minimize overheads.  Our
optimizations lead to throughput improvements of several hundred tokens per
second when hosting inference for LM agents. 
We provide a theoretical analysis of our algorithms to provide insights
into speculation configurations that will yield the best performance.
Further, we recommend a new ``tool cache'' API endpoint to enable LM providers to easily adopt these
optimizations.

\end{abstract}

\section{Introduction}\label{sec:intro}
Tool and function calling have enabled language models (LMs) to become useful
for tasks beyond just conversation by providing the ability to interact with
external environments and collect further context~\cite{yao2022react, schick2023toolformer, qin2024tool, qu2025tool}. In particular, LM-based
agentic tools and frameworks are often entirely reliant on external tools as
they are designed to interact with the environment to solve some problem or
accomplish a task. Popular agents such as software engineering agents (SWE
agents)~\cite{yang2024swe, liu2024large} need to interact with source code files and the command line to
execute actions. Although access to external tools yields much richer
capabilities and enables LMs to solve long-horizon, real-world tasks, it also
introduces several performance bottlenecks in the traditional inference
pipeline. Instead of a single, contiguous generation, the model alternates
between generation and tool invocation, often across many concurrent sessions,
resulting in gaps due to waiting on tool completion.

\begin{figure}[h]
    \centering
    \includegraphics[width=0.95\linewidth]{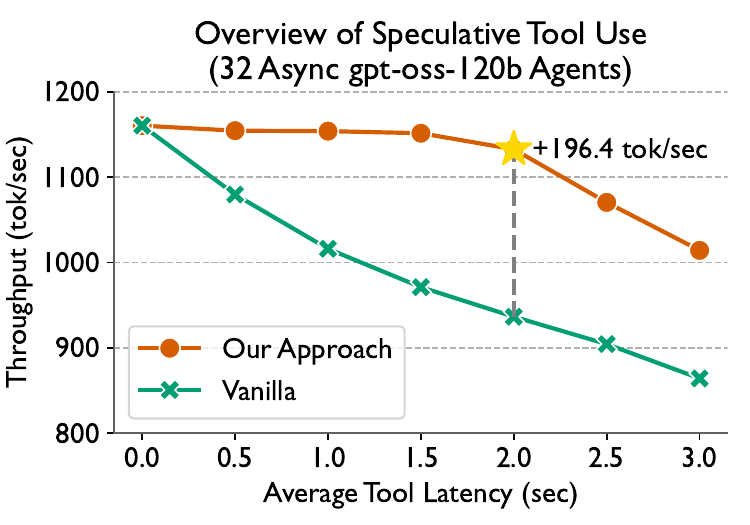}
    \caption{Our approach leads to up to 196 tokens/second improvement in the vLLM server when there are 32 gpt-oss-120b agents using it for inference. \label{fig:results-intro}}
\end{figure}

As tool-centric agents become more prevalent in code co-pilots, personal
assistants, and autonomous workflows, the overhead of using these agents is
increasingly dominated by the time spent waiting on tools.  Each tool call forces generation to stop and return to the user, who handles running the tool and sending back its output. This interrupt-driven process introduces a strict sequential dependency, where the latencies of individual tool calls accumulate and can significantly increase total generation time. Additionally, the evicted sequences and tool output must be rescheduled back into the inference engine, causing further overhead. Even with optimizations such as prefix-caching~\cite{kwon2023efficient, vllm_apc, zheng2024sglang, gao2024cost}, there are still significant overheads to removing the sequences, processing the tool output, and rescheduling them for generation. These overheads are further exacerbated in multi-tenant settings where scheduling and prefix-caching optimizations become more challenging with many agents contending for resources, and many concurrent prompts evicting existing entries in the KV-cache. Removing the sequential dependence and reducing these overheads in tool-heavy workloads is critical as agentic LMs become more widespread and need to be more economically viable and efficient.

Breaking the strict sequential dependency introduced by tool calls and reducing eviction overheads is non-trivial from a systems perspective.
Tool latencies can span orders of magnitude and depend
on external services, filesystem I/O, or user-defined code, making them
difficult to predict or bound.  
Long-running tools inevitably dominate end-to-end latency, whereas for short tools, the overheads of eviction and re-entry into the engine can outweigh any potential latency hiding strategies.
Furthermore, decoupling tool execution from model progression introduces correctness challenges, and naively running tools early or out of order can lead to wasted computation with unused results or outputs that are inconsistent with the model’s eventual decisions.
Finally, existing inference optimizations such as speculative decoding~\cite{leviathan_fast_2023, chen2023accelerating} and optimized KV-cache management~\cite{kwon2023efficient}
are implemented entirely within the inference engine and treat tools as opaque
interruptions, making it difficult for client-side innovations to impact
engine performance.

We address these challenges in agent tool calling by introducing \emph{speculative tool execution}, realized through two variants. First, we present a client-side strategy that breaks the sequential nature of tool execution without requiring any modifications to the inference engine, making it immediately compatible with existing inference services. Second, we introduce an engine-level mechanism that enables the inference system to ingest tool outputs directly and eliminate eviction-related overheads. These strategies allow us to mask portions of tool execution with generation for stateless and cheap-to-execute tools.
Furthermore, we provide an analytical performance model for both approaches and discuss the impact of tool durations, prediction accuracies,
and decode latencies on the performance improvements. The proposed
approaches are evaluated across a range of workloads, demonstrating up to
several hundred tokens per second throughput improvements.
\Cref{fig:results-intro} highlights the savings of our approach over standard inference with tool calling agents.

In this work, we make the following important contributions:
\begin{itemize}
    \item Two novel methods for speculating tool calls to optimize inference for agents by speculating work and keeping prompts resident in the engine.
    \item A theoretical analysis on the limitations of our methods, and which configurations will yield highest performing results.
    \item A detailed study of our two methods on agent workflows with tool calling.
\end{itemize}

\section{Background}\label{sec:bg}

In this section we provide background on tool calling and speculative decoding.

\subsection{Tool Calling}\label{sec:bg-tool-calling}

Most modern LMs use \emph{tool calling} (or \emph{function calling}) as an interface for interacting with external systems (command line, APIs, etc).
In a typical setup, the user provides (1) a natural language prompt describing the task and (2) a list of available tools with structured argument schemas.
Given this input, the model responds with either natural language or a structured \emph{tool call} specifying which tool to invoke and with what arguments.
The client then executes the tool and sends the tool output back to the model that includes the original conversation plus the tool output.
In LM agents this process is typically repeated until the model has enough context to produce a final output or it has completed its tasks.

From the perspective of the inference engine, each tool call breaks the generation into multiple, shorter requests separated by externally executing the tool. For any request, the engine first processes all tokens in the prompt during the \emph{prefill} phase to generate the first output token, and then produces subsequent tokens auto-regressively during the \emph{decode} phase. When a tool call is generated, further decoding of the sequence is halted, and it is removed from the active batch while the client executes the tool. Once the tool result is available, the client submits a new request consisting of the original prompt, the intermediate text, and the tool output, forcing the engine to prefill the history again before it can decode the next tool call or final output.
For agentic workloads with many tools and long multi-step plans, this leads to fragmented generations and overheads from repeatedly restarting generation with prompts increasing in length.
In long agent executions prefilling long prompts and waiting on tools can easily become big bottlenecks.

Two common optimizations are used today to mitigate these costs:

\vspace{0.08in}
\noindent \textbf{Parallel tool calls.}
    Instead of emitting a single tool invocation per turn, models can return multiple independent tool calls in one response~\cite{singh2024llm, chen2024facilitating}.
    This enables the client to execute tools concurrently and then return results back to the model in a single response.
    This reduces the number of model/tool round-trips, but still introduces gaps during inference where the engine has to hand off to the client.
    Furthermore, sometimes tools need to be called sequentially as they have a dependence on each other or the model needs the output of one tool before it can decide to call the next.
    Models also need to be trained to support parallel tool calling.

\vspace{0.08in}
\noindent \textbf{Prefix caching.}
    Many inference runtimes maintain a \emph{prefix cache} of KV state for recently used prompts~\cite{zheng2024sglang, kwon2023efficient}.
    This enables the runtime to reuse the KV cache if a user submits another request with the same history as a request they submitted shortly before; prefix caching drastically reduces prefill overhead during quick multi-turn API use.
    However, it does not eliminate the overhead of removing sequences from the batch, re-scheduling them, or managing cache eviction under multi-tenant loads, which can actually have high overheads if privacy-preserving algorithms are required~\cite{vllm_apc, song2025early}.

\subsection{Speculative Decoding}\label{sec:bg-spec-decoding}
Speculative decoding is an inference optimization that accelerates LM generation by utilizing a smaller, faster \emph{draft} LM to speculate tokens for a larger, slower LM.
Given the same prompt, the draft model proposes several \emph{speculative} tokens ahead of the target model.
The target model then verifies these proposed tokens in parallel: as long as the draft tokens are likely under the target’s distribution, they are accepted; when a mismatch occurs, the algorithm falls back to sampling from the target model at the first rejection point.
This was first proposed by Leviathan et al.~\cite{leviathan_fast_2023} and has been modified and improved in many ways since its first introduction~\cite{hu_speculative_2025}.

This setup works well for two reasons: draft models can often predict some easier text with high accuracy and the target model can validate tokens in a single forward pass.
The latter means there is no slowdown to inference, a forward pass of the target model is run regardless to sample the next token, but if it validates the speculative model, then we get several tokens for free.
The only extra cost is the compute needed to run the speculative model.

This work adapts the ideas behind speculative decoding, but with some subtle differences.
Generally, draft models in speculative decoding are only 3-10 tokens ahead of the target model.
This works great for speculative generation, but in speculative tool calling, if we were to launch tools as the draft model generates them, then we would only be overlapping tool calls with a few generated tokens.
This would only save a few milliseconds at most and be largely dominated by the longer tool calling times and API overheads.
Instead, with tools it is better to speculate as soon as possible and, in the case of stateless cheap tools, as many times as possible.
Our methodology is designed around this insight and key difference with speculative decoding.

\section{Speculative Tool Calling Algorithms}\label{sec:methodology}
Our goal in this work is to optimize the performance of tool calling in LM-based agents by speculating tool calls and executing them before they are needed. We propose two approaches for speculative tool calling: strictly client-side speculation that reduces tool latency and an inference-engine-side speculative algorithm that reduces both decode and prefill time.

In traditional tool calling, the user first submits a prompt to an API endpoint inference server.
The server starts processing the prompt and then returns the generated tokens to the user when it hits a stop token.
This could be the end of the sequence or a tool call token.
When the user receives the response from the server, they check why the generation stopped (e.g. tool call, end of sequence, timeout, max tokens, API credits) and, if it is a tool call, then they run the tool locally using the tool parameters generated by the model.
Once the tool completes the user passes the output back to the API and the LM can keep generating using the tool output.
This process is detailed in \Cref{alg:baseline-tool}.

\begin{algorithm}[h]
\caption{Baseline Tool-Calling (Client-side, No Speculation)}
\label{alg:baseline-tool}
\begin{algorithmic}[1]
    \Require Prompt $p$, chat history $H$, tools $\mathcal{T}$, main LM $\mathcal{M}$
    \Ensure Final model output $y$
    \State $r \gets \textsc{CallAPI}(\mathcal{M}, p, H, \mathcal{T})$
    \If{$\textsc{HasToolCall}(r)$}
      \State $\tau \gets \textsc{ExtractToolCall}(r)$
      \State $v \gets \textsc{RunTool}(\tau)$
      \State $H' \gets H \cup \{(\tau, v)\}$ \Comment{append tool result to history}
      \State $y \gets \textsc{CallAPI}(\mathcal{M}, p, H', \mathcal{T})$ \Comment{return final answer}
    \Else
      \State $y \gets \textsc{Render}(r)$ \Comment{no tool call; use first pass}
    \EndIf
    \State \Return $y$
\end{algorithmic}
\end{algorithm}

\subsection{Client-side Speculative Tool Calling}\label{sec:spec-tool:user}

\begin{figure*}[t]
    \centering
    \includegraphics[width=0.8\textwidth]{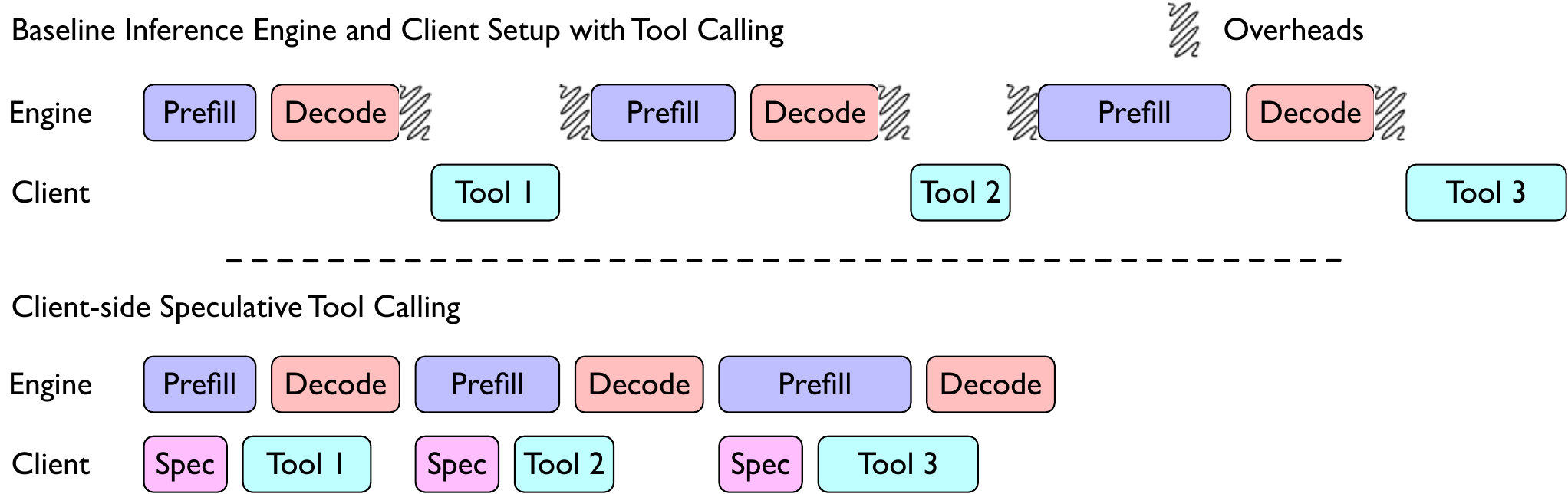}
    \caption{Overview of the client-side approach for speculative tool calling. Our approach speculates tool calls and overlaps their execution with generation.}
    \label{fig:client-side}
\end{figure*}

We expand on the traditional tool calling approach with speculative tool calling on the client side as described in \Cref{alg:spec-tool-cache}.
In this approach we asynchronously send our prompt to two API endpoints: the main model $\mathcal{M}$ and a speculative model $\mathcal{S}$.
The speculative model is much smaller and should return estimated tool calls with much lower latency than the main model.
When $\mathcal{S}$ returns a tool call, we asynchronously launch the tool and start executing it.
We store a {\it future} for its result in a tool cache.
When $\mathcal{M}$ finally returns generated tokens, we check if a tool call is required.
If it is required and there is a future for that tool call in the cache, we wait on the future instead of relaunching the tool. Figure~\ref{fig:client-side} provides an overview of our approach.
We can further increase speculation by launching several speculative model instances (represented by $\lambda$ in \Cref{alg:spec-tool-cache}).

\begin{algorithm}[!htbp]
\caption{Client-side Speculative Tool-Calling}
\label{alg:spec-tool-cache}
\begin{algorithmic}[1]
    \Require Prompt $p$, chat history $H$, tools $\mathcal{T}$, sampling factor $\lambda$, main LLM $\mathcal{M}$, speculative LM $\mathcal{S}$
    \Ensure Final model output $y$
    \State $C \gets \textsc{EmptyMap}()$ \Comment{cache: \textit{tool call} $\mapsto$ \textit{future}}
    \State \textbf{spawn} $h_m \gets \textsc{CallAPI}(\mathcal{M}, p, H, \mathcal{T})$ \Comment{start main generation asynchronously}
    \For{$i \gets 1$ to $\lambda$} \Comment{speculative samples in parallel}
      \State \textbf{spawn}:
      \Statex \hspace{1.2em} $s_i \gets \textsc{CallAPI}(\mathcal{S}, p, H, \mathcal{T})$
      \If{$\textsc{HasToolCall}(s_i)$}
        \State $\tau \gets \textsc{ExtractToolCall}(s_i)$
        \If{$\tau \notin C$}
          \State $C[\tau] \gets \textsc{StartToolAsync}(\tau)$ \Comment{begin executing tool; returns a future}
        \EndIf
      \EndIf
    \EndFor
    \\
    \State $r_m \gets \textbf{await}\; h_m$ \Comment{wait for main generation}
    \\
    \If{$\textsc{HasToolCall}(r_m)$}
      \State $\tau_m \gets \textsc{ExtractToolCall}(r_m)$
      \If{$\tau_m \in C$}
        \State $v \gets \textbf{await}\; C[\tau_m]$ \Comment{reuse if available}
      \Else
        \State $v \gets \textsc{RunTool}(\tau_m)$ \Comment{execute tool now}
      \EndIf
      \State $H' \gets H \cup \{(\tau_m, v)\}$ \Comment{append tool result to history}
      \State $y \gets \textsc{CallAPI}(\mathcal{M}, p, H', \mathcal{T})$ \Comment{call $\mathcal{M}$ with result}
    \Else
      \State $y \gets \textsc{Render}(r_m)$ \Comment{no tool call}
    \EndIf
    \\
    \State \Return $y$
\end{algorithmic}
\end{algorithm}

It is important to note that we can only speculate \emph{stateless, cheap} tools.
If a tool or function requires or modifies state, then speculating it is impossible without some undo or rollback mechanism.
This is an interesting potential optimization but out of the scope of this paper.
Further, if a tool call is expensive, either in dollar costs or execution time, it is undesirable to launch many attempts of the tool ahead of time if they will not be correct.
Fortunately, many common agent tools, such as web-search and file access, are both cheap and stateless providing ample opportunity for speculation.

\subsubsection{Analysis of Client-side Speculation}\label{sec:client-side-analysis}

We examine a particular case where the user is repeatedly calling the API and receiving a tool call from the API each time.
This closely mimics the usage pattern of most AI agents and is, thus, a reasonable case to consider.
The times to process $N$ requests from the user in the standard (\Cref{alg:baseline-tool}) and speculative (\Cref{alg:spec-tool-cache}) case are shown below.

\vspace{1.5em}
\begin{align}
    T_{\mathtt{standard}} &= N (\eqnmarkbox[MidnightBlue]{G1}{G} + \eqnmarkbox[WildStrawberry]{T1}{T}) \nonumber \\
    T_{\mathtt{spec}} &= \eqnmarkbox[Orchid]{a1}{\alpha} N \max\{\eqnmarkbox[MidnightBlue]{G2}{G},\eqnmarkbox[OliveGreen]{g}{g}+\eqnmarkbox[WildStrawberry]{T2}{T}\} + (1-\eqnmarkbox[Orchid]{a2}{\alpha})N(\eqnmarkbox[MidnightBlue]{G3}{G}+\eqnmarkbox[WildStrawberry]{T3}{T})
\end{align}
\annotate[yshift=1em]{above,left}{G1}{avg. $\mathcal{M}$\\ generation time}
\annotate[yshift=1em]{above,right}{T1}{avg. tool call time}
\annotate[yshift=0.5em,xshift=0.3em]{above,right}{g}{avg. $\mathcal{S}$ generation time}
\annotatetwo[yshift=-0.8em]{below}{a1}{a2}{$\mathcal{S}$ acceptance rate}
\vspace{1em}

The left term in $T_{\mathtt{spec}}$ is the case where the speculative model was correct and, thus, we only wait $g+T$ for the tool to complete.
The maximum is necessary since we still need to wait on $\mathcal{M}$ to complete to validate the right tool call, so if $g+T<G$, then we still have to wait $G$.
Assuming $T$, $g$, $G$, and $N$ are positive, non-zero and $0\le\alpha\le1$ we can further analyze the speedup of the speculative algorithm.
\begin{align}
    S_{\mathtt{spec}} =& \frac{T_{\mathtt{standard}}}{T_{\mathtt{spec}}} &> 1 \nonumber \\
  =& \frac{G + T}{\alpha \max\{G, g+T\} + (1-\alpha)(G+T)} &> 1 \label{eq:client-alg-speedup}
\end{align}
\begin{align}
  \text{If $g+T \ge G$, then} \nonumber \\
  S_{\mathtt{spec}} =& \frac{G+T}{\alpha(g-G)+G+T} &> 1 \nonumber \\
  & \alpha(g-G)+G+T &< G+T  \nonumber \\
  & \alpha(g-G) &<0\nonumber \\
    \implies & \alpha > 0 \textrm{ and } g < G \label{eq:client-alg-condition}
\end{align}

We can immediately see that $S_{\mathtt{spec}}>1$ when $\mathcal{S}$ has a non-zero acceptance rate and $\mathcal{S}$ is faster than $\mathcal{M}$. 
The case where $g+T<G$ is ignored, since this trivially leads to speedups greater than 1 when $\alpha>0$ (i.e. speculation and tool calling are completely masked by the main models generation time, so any correct speculation should lead to speedups).
Based on \Cref{eq:client-alg-condition} \emph{to achieve speedups with speculative tool calling on the client side we need to find a speculation model $\mathcal{S}$ that is faster than the main model $\mathcal{M}$ and has a tool speculation accuracy greater than 0}.

Another key insight from \Cref{eq:client-alg-speedup} is that $S_{\mathtt{spec}} < 2$, i.e. \emph{the maximum speedup we can achieve in this approach is 2.
Furthermore, the speedup approaches 2 as $g\rightarrow 0$}.
\Cref{fig:client-alg-model-speedups} shows the distribution of $S_{\textrm{spec}}$ across values of $\alpha$, $g/G$, and $T$.
We see the above trends, i.e. that $S_{\textrm{spec}}\rightarrow 2$ as $g\rightarrow 0^+$ and $S_{\textrm{spec}}<2$, but also that we only see large speedups for values of $T$ near $G$.
As $T\rightarrow 0$ the total time is predominantly dominated by $G$ and as $T\rightarrow \infty$ it is dominated by the tool call.
Thus, when $T\approx G$ there are more values of $\alpha$ and $g/G$ where we can see substantial benefits.
The above insights are also proven below in \Cref{lemma:client-alg-proof}.

\begin{figure*}[t]
    \centering
    \includegraphics[width=0.9\textwidth]{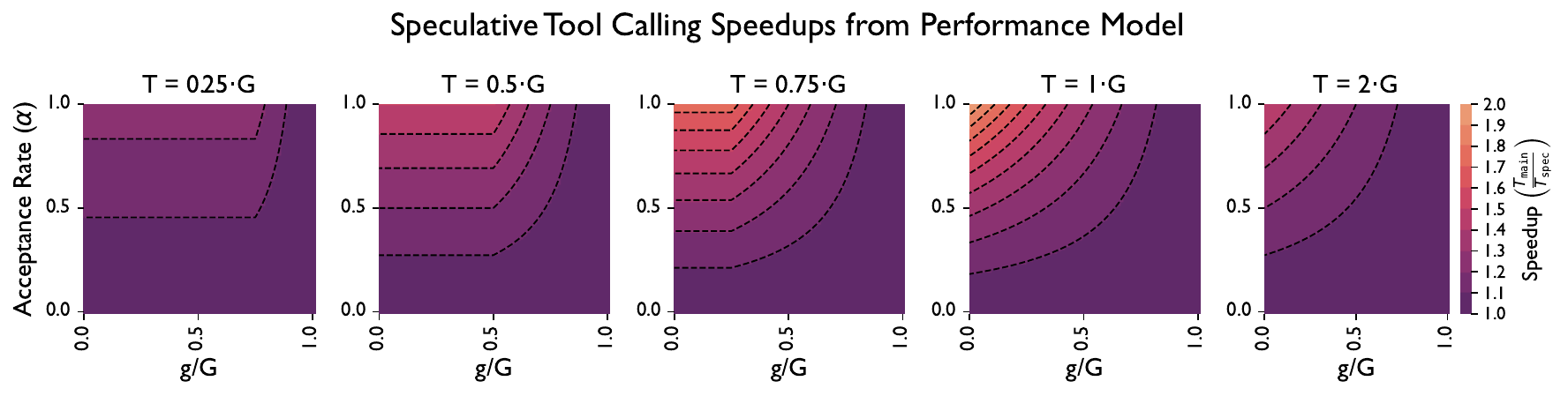}
    \caption{Distribution of speculative tool calling speedups for the client-side algorithm across various acceptance rates, generation times, and tool call times. $\alpha$ is the acceptance rate of the speculative model $\mathcal{S}$, $g/G$ is the ratio of generation time between the speculative and main models, and $T$ is the tool call time. Values of $T\approx G$, $\alpha > 0.5$, and $g/G < 0.5$ yield higher speedups, but are ultimately limited at a 2$\times$ speedup. \label{fig:client-alg-model-speedups}}
\end{figure*}

\begin{lemma}[Speedup bound] \label{lemma:client-alg-proof}
Let $G,g,T>0$, $g<G$, and $\alpha\in[0,1]$. Define
\[
S_{\mathrm{spec}}(\alpha)
=\frac{G+T}{\alpha\,\max\{G,\,g+T\}+(1-\alpha)(G+T)}.
\]
Then:
\begin{enumerate}
  \item $S_{\mathrm{spec}}(\alpha)$ is strictly increasing in $\alpha$, hence $S_{\mathrm{spec}}(\alpha)>1$ for all $\alpha>0$.
  \item The maximum over $\alpha\in[0,1]$ is attained at $\alpha=1$ and equals
  \[
  S_{\max}=\frac{G+T}{\max\{G,\,g+T\}}.
  \]
  \item Moreover,
  \[
  S_{\max}\le \frac{2(G+T)}{G+g+T}
  = 2-\frac{2g}{G+g+T} \;<\; 2.
  \]
  In particular, the supremum $2$ is approached only in the limit $g\to 0^+$.
\end{enumerate}
\end{lemma}

\begin{proof}
Let $M:=\max\{G,\,g+T\}$ and write
\begin{align}
S_{\mathrm{spec}}(\alpha) =& \frac{G+T}{\alpha M+(1-\alpha)(G+T)} \nonumber \\
=& \frac{G+T}{(G+T)+\alpha\,(M-(G+T))}. \nonumber 
\end{align}
Set $\Delta:=M-(G+T)$. Since $M\le G+T$, we have $\Delta\le 0$, with strict inequality because $G,g,T>0$.

\smallskip
\noindent\textbf{(1)} Since $M\le G+T$ and $G,g,T,\alpha>0$, the denominator of $S_{\mathrm{spec}}(\alpha)$ is strictly decreasing as $\alpha$ increases.
Thus, we can conclude that $S_{\mathrm{spec}}(\alpha)$ is strictly increasing with $\alpha$. Furthermore, since $S_{\mathrm{spec}}(0)=1$ it follows that $S_{\mathrm{spec}}(\alpha)>1$ for all $\alpha>0$.

\smallskip
\noindent\textbf{(2)} Monotonicity implies the maximum over $[0,1]$ occurs at $\alpha=1$, giving
\[
S_{\max}=S_{\mathrm{spec}}(1)=\frac{G+T}{M}=\frac{G+T}{\max\{G,\,g+T\}}.
\]

\smallskip
\noindent\textbf{(3)} For any $a,b>0$, $\max\{a,b\}\ge \tfrac{a+b}{2}$. 
Applying this with $a=G$ and $b=g+T$,
\begin{align}
S_{\max} =& \frac{G+T}{\max\{G,\,g+T\}} \nonumber \\
\le& \ \frac{G+T}{\tfrac{G+(g+T)}{2}}
= \frac{2(G+T)}{G+g+T}
= 2-\frac{2g}{G+g+T} \nonumber \\
<& \  2, \nonumber
\end{align}
where the strict inequality uses $g>0$. This proves the stated bound and its strictness.
\end{proof}

\begin{observationbox}{Client-side speculation limit}
Client-side speculative tool calling is limited at a $2\times$ speedup: we can only
hide one of the two dominant phases (generation or tool execution). 
For a good speculative model (fast and accurate), the best gains occur when the tool latency $T$ is approximately equal to the main model generation time $G$.
\end{observationbox}

\subsection{Engine-side Speculative Tool Calling}\label{sec:spec-tool:prefill}

\begin{figure*}[t]
    \centering
    \includegraphics[width=0.8\textwidth]{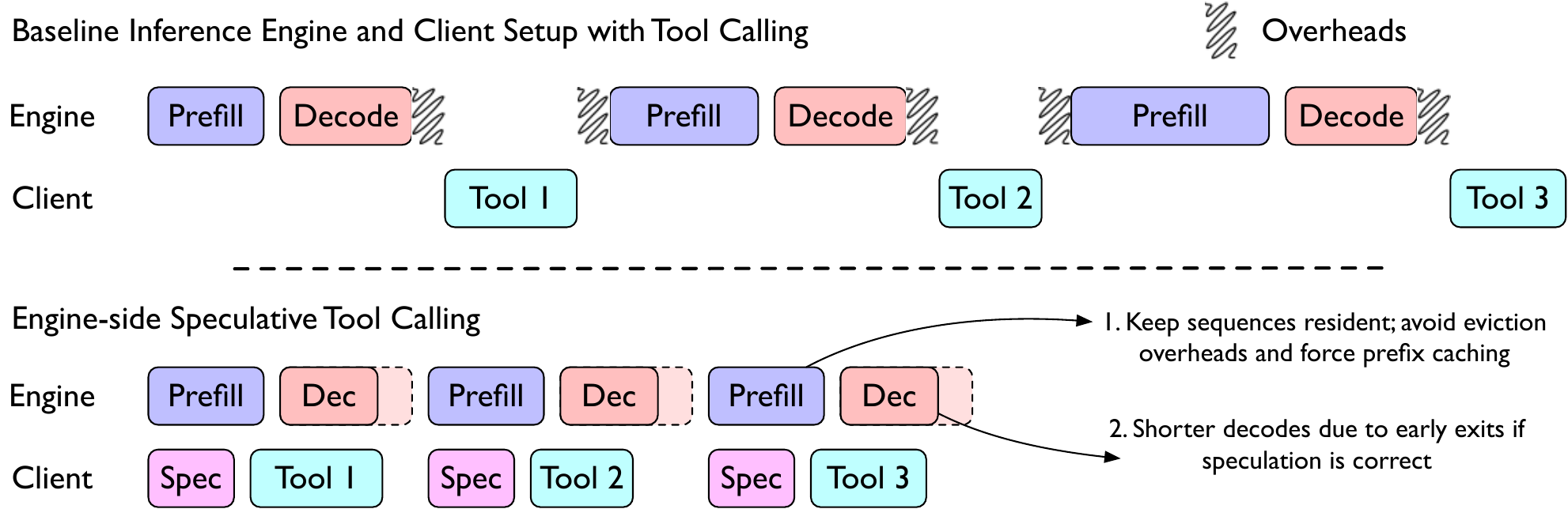}
\caption{Overview of our proposed engine-side algorithm for speculative tool calling for agents. First, our approach speculates tool calls and overlaps their execution with generation. Second, if speculated tool calls are available before the target model finishes decoding it, then we can validate the tool call in a single forward pass and early exit decoding. Finally, by posting speculated tool outputs to the server, we prevent the prefill time from continuing to grow and remove overheads from removing the prompt from the batch.}
    \label{fig:engine-side}
\end{figure*}

Next, we present our approach of moving our speculative tool calling method from the client to the engine. First, we present how a non-speculative inference engine with tool calling works in 
Algorithm \ref{alg:engine-baseline}. The engine repeatedly ingests new requests from the input queue into an active batch. For each sequence in the batch, it first encodes the prompt to produce the first token (prefill) and then generates subsequent tokens (decode). When a stop token (end of sequence or tool call token) is produced, the engine stops generation. If the stop token is a tool call token, the engine returns the tool call to the client, and evicts the sequence from the active batch. Its KV-cache is discarded unless prefix-caching is enabled.

\begin{algorithm}[!htbp]
\caption{Baseline Inference Engine: Batch Decode with Evict-then-Refill on Tool Calls}
\label{alg:engine-baseline}
\begin{algorithmic}[1]
\Require Input queue $\mathsf{Q}$, model $\mathcal{M}$, batch size $K$
\State $\mathsf{B} \gets \emptyset$ \Comment{active batch of sequences}
\While{\textsc{ServiceRunning}()}
  \State \textbf{fill} $\mathsf{B}$ from $\mathsf{Q}$ up to $K$ ready requests
  \ForAll{$s \in \mathsf{B}$}
  \If {$s.\textsc{status} = \textsc{New}$}
    \State $t \gets \textsc{Prefill}(\mathcal{M}, s)$ \Comment{first token}
    \State $s.\textsc{status} \gets \textsc{Decode}$
    \ElsIf {$s.\textsc{status} = \textsc{Decode}$}
  \State $t \gets \textsc{DecodeStep}(\mathcal{M}, s)$ \Comment{get next token}
    \EndIf
    \State
      \If{$\textsc{IsStop}(t)$}
        \State \textsc{EmitToClient}$(s, \text{final})$
        \State \textsc{Evict}$(\mathsf{B}, s)$ \Comment{KV-cache maybe dropped}
      \EndIf
    \EndFor
  \State \textbf{ingest} any new/follow-up requests from client (e.g. tool results) into $\mathsf{Q}$
\EndWhile
\end{algorithmic}
\end{algorithm}

The engine-side speculative approach adds a tool cache maintained by the inference engine. Similar to the previous approach, the client spawns multiple speculative model instances and launches tools asynchronously. However, instead of storing futures, it waits for each tool execution to finish and then submits the result to the engine, indexed by a normalized key (tool name + canonicalized arguments) and the request ID. On the engine side, we introduce two optimizations based on the tool cache. When a tool start token (the first token indicating that the model is about to produce a tool invocation) is detected for a sequence, the engine performs a cache lookup using the request ID and tool name.  
If found, the argument tokens of the latest matching tool call are injected as draft tokens and validated through standard speculative sampling. This ensures correctness and potentially avoids the need to generate the tool call arguments (early exit decoding). When a tool end token (the token indicating that the full tool call has been emitted) is detected, the engine looks up the full key. On a cache hit, the tool result is appended and the KV-cache updated so decoding can continue without eviction. Otherwise, the engine falls back to baseline behavior and emits the tool call to the client.~\Cref{fig:engine-side} provides an overview of this approach. The client and engine behaviors under this approach are described in Algorithms~\ref{alg:client-speculative} and~\ref{alg:engine-speculative}, respectively.

\begin{algorithm}[h]
\caption{Client with Speculative Tool Execution and Cache Submission}
\label{alg:client-speculative}
\begin{algorithmic}[1]

\Require Prompt $p$, chat history $H$, tools $\mathcal{T}$, main LM API endpoint $\mathsf{E}$, speculative LLM $\mathcal{S}$, samples $\lambda$
\Ensure Final output $y$
\State $(h, \mathit{rid}) \gets \textsc{CallAPI}(\mathsf{E}, p, H, \mathcal{T})$ \Comment{start main request; get request id}
\For{$i \gets 1$ to $\lambda$} \Comment{launch speculative samples}
  \State \textbf{spawn} $s_i \gets \textsc{CallLLM}(\mathcal{S}, p, H, \mathcal{T})$
  \If{$\textsc{HasToolCall}(s_i)$}
    \State $\tau \gets \textsc{ExtractToolCall}(s_i)$
    \State $f \gets \textsc{StartToolAsync}(\tau)$
    \State \textbf{spawn} \textsc{OnReady}$(f)$:
    \Statex \hspace{4em} $v \gets \textbf{await}\; f$
    \Statex \hspace{4em} $k \gets \textsc{CanonKey}(\tau)$ \Comment{normalize: tool name + canonicalized args}
    \Statex \hspace{4em} \textsc{SubmitToolCache}$(\mathsf{E}, \mathit{rid}, k, v)$
  \EndIf
\EndFor
\State
\State $r \gets \textbf{await}\; h$ \Comment{receive result (may be final or a tool call)}
\State
\If{$\textsc{HasToolCall}(r)$}
  \Comment{cache miss path: engine did \emph{not} find value in its cache}
  \State $\tau^\star \gets \textsc{ExtractToolCall}(r)$
  \State $v^\star \gets \textsc{RunTool}(\tau^\star)$
  \State $H' \gets H \cup \{(\tau^\star, v^\star)\}$
  \State $r \gets \textsc{CallAPI}(\mathsf{E}, p, H', \mathcal{T})$ \Comment{continue normal loop}
  \While{$\textsc{HasToolCall}(r)$}
    \State $\tau \gets \textsc{ExtractToolCall}(r)$
    \State $v \gets \textsc{RunTool}(\tau)$; $H' \gets H' \cup \{(\tau, v)\}$
    \State  $r \gets \textsc{CallAPI}(\mathsf{E}, p, H', \mathcal{T})$
  \EndWhile
\EndIf
\State $y \gets \textsc{Render}(r)$
\State \Return $y$
\end{algorithmic}
\end{algorithm}

\begin{algorithm}[h!]
\caption{Inference Engine with Tool Cache} %
\label{alg:engine-speculative}
\begin{algorithmic}[1]
\Require Input queue $\mathsf{Q}$, model $\mathcal{M}$, batch size $K$, tool cache $\mathsf{C}$ keyed by (request id, canon key) and (request id, tool name)
\State $\mathsf{B} \gets \emptyset$
\While{\textsc{ServiceRunning}()}
  \State \textsc{AsyncUpdate}($\mathsf{C}$) \Comment{from \textsc{SubmitToolCache}}
  \State \textbf{fill} $\mathsf{B}$ from $\mathsf{Q}$ up to $K$
  \ForAll{$s \in \mathsf{B}$}
  \If {$s.\textsc{status} = \textsc{New}$}
    \State $t \gets \textsc{Prefill}(\mathcal{M}, s)$; $s.\textsc{status} \gets \textsc{Decode}$
    \ElsIf{$s.\textsc{status} = \textsc{Decode}$}
      \State $t \gets \textsc{DecodeStep}(M, s)$
    \EndIf
    \State
     \If{$\textsc{IsToolCallStart}(t)$}
        \State $\eta \gets \textsc{ExtractToolName} (s, t)$
        \If{$\mathsf{C}.\textsc{Has}(\textit{rid}(s),\eta)$}
        \State $c \gets \mathsf{C}.\textsc{GetToolCall}(\textit{rid}(s),\eta)$
        \State $t \gets \textsc{SpeculativeSample}(\mathcal{M}, s, c)$ \Comment{Validate tool call tokens and update KV cache}
        \EndIf
        \ElsIf{$\textsc{IsToolCallEnd}(t)$}
        \State $\tau \gets \textsc{ExtractToolCall}(s, t)$
        \State $k \gets \textsc{CanonKey}(\tau)$
        \If{$\mathsf{C}.\textsc{Has}(\textit{rid}(s), k)$} \Comment{cache hit}
          \State $v \gets \mathsf{C}.\textsc{GetToolResult}(\textit{rid}(s), k)$
          \State $t \gets \textsc{PrefillToolResult}(\mathcal{M}, s, \tau, v)$ \Comment{append tool result tokens to KV cache}
        \EndIf
        \EndIf
      \State
      \If{$\textsc{IsStop}(t)$} \Comment{cache miss/other stop token}
        \State \textsc{EmitToClient}$(s, \text{final})$ 
        \State \textsc{EvictFromBatch}$(\mathsf{B}, s)$
      \EndIf
    \EndFor
  \State \textbf{ingest} \textsc{New} client requests into $\mathsf{Q}$
\EndWhile
\end{algorithmic}
\end{algorithm}

\subsubsection{Analysis of Engine-side Speculation}\label{sec:engine-side-analysis}

In the traditional case of tool calling, where the user receives the tool parameters, runs it, and calls the API again, we can model the runtime from the inference engine's perspective as shown in \Cref{eq:spec-engine-vanilla} for $K$ consecutive agent turns.
This accounts for the prefill and decode phase of each turn where $X_i$ tokens are prefilled, $R_i$ reasoning tokens are decoded, and $t_i$ tool call tokens are decoded.
Then the tool is called for $T_i$ seconds after which the tool call and output ($t_i+t_{o,i}$ tokens) are sent back to the API.
Finally, the entire history and new tool calls will be prefilled again, so we sum them in the $X_i$ recursive term.
In between these phases there is also overhead, $o$, to passing the prompts through the API via HTTP and loading it into the running batch.

\vspace{1em}
\begin{equation}\label{eq:spec-engine-vanilla}
    T_{\textrm{vanilla}} = 
    2K\eqnmarkbox[MidnightBlue]{o1}{o} + 
    \eqnmarkbox[Apricot]{p1}{\varphi} \sum_{i=1}^K X_i + 
    \eqnmarkbox[OliveGreen]{d1}{\delta} \sum_{i=1}^K 
        \left(
            \eqnmarkbox[Orchid]{r1}{R_i} + 
            \eqnmarkbox[Mahogany]{t1}{t_i}
        \right) + 
    \sum_{i=1}^K \eqnmarkbox[Red]{T1}{T_i}
\end{equation}
\annotate[yshift=1.2em,xshift=0em]{above,left}{o1}{API and eviction\\overheads}
\annotate[yshift=-1.2em]{below,left}{p1}{Prefill rate (s/tok)}
\annotate[yshift=1.2em]{above, left}{d1}{Decode rate (s/tok)}
\annotate[yshift=1.2em]{above, left}{T1}{Tool duration (s)}
\annotate[yshift=-1.2em]{below, left}{r1}{Reasoning tokens}
\annotate[yshift=-1.2em]{below, right}{t1}{Tool call tokens}
\vspace{1em}

such that

\begin{equation}
    X_{i+1} = X_i + t_i + t_{o,i}, \quad X_1 = \textrm{initial prompt tokens} \nonumber
\end{equation}

Here we assume that (1) reasoning tokens from previous turns are not included in the prefill for future turns as is typical for most commercial APIs and (2) 
there are no additional outputs alongside the tool call during each turn. The former assumption is generally true across API models, while the latter depends on 
how the agent framework is implemented (directly using tools versus thought-actions as in SWE-Agent). Despite this, our model is without loss of generality,
since you can absorb these extra decode tokens into the $t_i$ term.

Note that in the case of prefix-caching we can express the time as shown in \Cref{eq:spec-engine-pref-cache}. Here we do not have to repeat prefills for
tokens we have already populated into the KV-cache.

\begin{equation}\label{eq:spec-engine-pref-cache}
    T_{\textrm{cached}} = 2Ko + \varphi\left(X_1 + \sum_{i=1}^{K}(t_i+t_{o,i})\right) + \delta\sum_{i=1}^{K}(R_i+t_i) + \sum_{i=1}^{K} T_i
\end{equation}

Next we consider the three primary optimizations: (1) speculating tools prior to execution (\emph{O1}), (2) validating cached tool calls with a single forward pass during decoding (\emph{O2}), and (3) avoiding eviction and duplicate prefilling for tool calls where the output is already available (\emph{O3}).
The first optimization (\emph{O1}) can save us up to $\sum_{i=1}^K T_i$ time as it can potentially mask the tool call time.
The \emph{O2} optimization can potentially save $\delta \sum_{i=1}^K (t_i-1)$ time as we reduce the number of decodes from $t_i$ for generating the tool call to $1$ for validating the tool call.
Finally, the \emph{O3} optimization saves $2(K-1)o$ overhead in the ideal case and effectively forces prefix-caching to reduce the number of prefill tokens.
Based on these optimizations we can write the expected time in terms of the speculation accuracy $\alpha$ as show in \Cref{eq:spec-engine}.

\begin{align}\label{eq:spec-engine}
    T^{\ast}_{\textrm{spec}} =& (1-\alpha)2Ko + \varphi\left(X_1 + \sum_{i=1}^{K}(t_i+t_{o,i})\right) + \nonumber \\
    +& \delta \left(\alpha K + \sum_{i=1}^K R_i + (1-\alpha)\sum_{i=1}^K t_i\right) + (1-\alpha)\sum_{i=1}^K T_i
\end{align}

We can immediately see this is equivalent to $T_{\textrm{cached}}$ for $\alpha=0$. When $\alpha=1$ we save $2Ko + \delta(\sum_{i=1}^K t_i - K) + \sum_{i=1}^K T_i$ time over prefix-cached inference, since we avoid $2Ko$ overheads, the tool latencies, and the extra decodes for the tool calls.
Speedups will be optimal when the tool call duration $T_i$ is less than $\delta R_i$ so that the tool output will be available for speculation (\emph{O2}).
As in \Cref{sec:client-side-analysis} the optimal speedups for a single turn will then come when $T_i\approx \delta R_i$ as we are masking the entire tool call and saving on the most decode steps.
However, any value $T_i<\delta R_i$ should yield runtime savings.
This is an ideal finding as many state-of-the-art reasoning models commonly used in agents have several seconds to minutes reasoning traces, meaning we can speculate and execute many common tools in this time-frame.

\section{Experimental Setup}\label{sec:setup}
With the two proposed algorithms defined we now detail how we implemented them and the experiments we ran to test their effectiveness.

\subsection{Algorithm Implementations}\label{sec:impl}

The client-side algorithm is implemented on top of the OpenAI API~\cite{noauthor_openaiopenai-python_2025} using its async client.
Whenever a chat completion or response is posted to the main model, we also post a chat completion/response to the speculative API endpoint asynchronously.
If using multiple speculative samples, then we send multiple asynchronous requests.
When the speculative model API calls return, we launch the speculated tools asynchronously and put their futures into a tool cache keyed on function name and arguments.
This is implemented as a custom Python interface using Python's built-in asyncio library.
When the main model returns we check the tool cache if it called a tool.
If there is a cache hit, then we wait on the future and use its result as the tool output.
Otherwise we call the tool as normal.

For the engine-side algorithm we modify a custom fork of the popular vLLM framework~\cite{kwon2023efficient} to implement our algorithm.
A tool cache endpoint is implemented as a standard HTTP POST API endpoint (see later API spec). Similar to the client-side setup, one or more speculative model endpoints run asynchronously in parallel with the main model (\emph{O1}). The results of the speculated tool calls are posted to the tool-cache endpoint. 
Building on vLLM’s speculative decoding infrastructure, we implement a tool proposer responsible for detecting tool-call boundaries in the generated token stream. When the start of a tool call is detected, the proposer drafts the tool call tokens matching the latest entry in the cache with the same tool name. These draft tokens are then validated using standard speculative sampling~\cite{leviathan_fast_2023}, maintaining correctness (\emph{O2}). When the end of a tool call is detected, the proposer performs a lookup in the cache using both the tool name and its arguments. On a cache hit, the corresponding tool-output tokens are drafted; in this case, all drafted tokens can be accepted directly, since the match is exact on both name and arguments (\emph{O3}). If the lookup results in a cache miss, generation stops, and the engine returns the tool call to the client for normal execution.
For the engine-side algorithm to work in multi-turn scenarios we ensure streaming is enabled, so the client knows when to start speculating new tools.

The custom API endpoint we use is detailed below. This simple interface is all that is needed for existing commercial endpoint providers to expose this optimization to users.
Providers are incentivized to use the \texttt{cache-tool-output} API as it can cut down on inference time and increase throughput.
End-user can be incentivized by reduced costs and faster turnaround times. The cost reduction is not necessarily implicit, however,
many providers already offer discounts on prompts that hit the prefix-cache and our proposed optimization would fit nicely within that existing pricing model.

\begin{apispecbox}{\,\texttt{POST} \small \texttt{ /cache-tool-output/\{response\_id\}}}

\textbf{Description.}  
Cache tool outputs that the inference engine can utilize to keep sequences resident after tool calls. 
Each entry is keyed by \texttt{name} and (optionally) \texttt{params} and is matched against future tool calls by the model in this response.

\medskip
\noindent\textbf{Request body.}  
JSON array of objects with the following fields:
\medskip

\small
\begin{tabularx}{\linewidth}{@{}p{0.26\linewidth} L@{}}
\toprule
\textbf{Field} & \textbf{Description} \\
\midrule
\texttt{name} &
(string, required) Unique tool name used as part of the cache key. \\[0.25em]

\texttt{params} &
(object, optional) Canonicalized tool parameters included in the cache key.
Calls with the same \texttt{name} and equivalent \texttt{params} share outputs. \\[0.25em]

\texttt{output} &
(any, required) Serialized tool output to be cached and reused when the model
emits a matching tool call. \\[0.25em]

\texttt{keep\_alive} &
(number, optional) Optional time-to-live for the cache entry (e.g., in seconds). \\
\bottomrule
\end{tabularx}
\normalsize

\medskip
\noindent\textbf{Response.}
{\small
\begin{verbatim}
200 OK
{
  "cached": <int>  // number of entries accepted
}
\end{verbatim}
}

\end{apispecbox}

\subsection{Testing Environment}\label{sec:test-env}

We ran our implementations and experiments on a single compute node with four 80GB NVIDIA A100 GPUs controlled by one AMD EPYC 7763 CPU with 64 physical cores each at 2.45 GHz.
Each GPU pair is connected by third generation NVLink with 25 GB/s per direction links and each GPU is connected to the CPU via PCIe 4.0.
We run on a variation of SUSE Enterprise Linux Server 15 SP5 with Python 3.12 and CUDA 12.9.
In our experiments, we run two vLLM inference servers: one for the main model and one for the speculative model. 
The main model $\mathcal{M}$ uses a single A100 for serving, while the speculative server hosts three data parallel instances of the speculative model $\mathcal{S}$ across the other three A100s.

\subsection{Experiments}\label{sec:experiments}

To evaluate our inference optimizations, we use the dataset of prompts and tools from the BFCL project~\cite{patil_gorilla_2023}: a benchmark used to evaluate LMs on their tool calling capabilities.
BFCL defines a large number of tools that span a wide variety of tasks.
As a leaderboard, BFCL presents results comparing how well different LMs perform at calling the correct tools.
This provides a great starting point for selecting small, accurate speculation models.

From the BFCL leaderboards, we selected the xLAM models~\cite{zhang_xlam_2024} as ideal candidates for speculative models.
They are Llama 1B, 3B, and 8B fine-tunings on tool datasets to predict tool calls.
During testing we found the 8B model to predict around 80\% of tool calls from BFCL correctly, which is in line with the results they present in their paper.
Since these models are small and accurate at predicting tools, we utilize them as speculative tool calling models.

BFCL defines pairs of input prompts and sets of tools that can be invoked for each prompt, but it does not provide concrete tool implementations (only specifications). Our goal in this work is to study the inference- and systems-level behavior of agents with tools.%
We therefore treat tools as black boxes with specified interfaces and latencies. For each BFCL tool, we pre-compute a representative output using a mix of LM and human authored responses, and store these outputs in an in-memory cache. During experiments, whenever the model calls a tool, the cached tool output is returned and its runtime is controlled via manually configured tool latencies.

Crucially, our algorithms and measurements depend only on \emph{when} tools are called and \emph{how long} they take to run, not on the content of the tool outputs themselves. Precomputing outputs in this way preserves the model’s compute behavior while allowing us to systematically control tool latencies. As a consistency check, we manually inspected model generations across our experiments and found them to be reasonable and in line with typical outputs from the base model. We also ran a small experiment using simple, hand-implemented tools with real executions and observed no qualitative change in model behavior or quantitative change in the performance results, supporting the validity of our tool setup for system-level performance evaluation.

To turn BFCL into an agent workload, we treat each prompt and toolset pair as an independent task for an asynchronous agent. We run $M$ concurrent agents with $M \in \{1, 8, 32\}$. 
Each agent samples a prompt and its associated tools. It will loop continually calling tools (when requested) and generating text until it is finished and produces a final response. Then it will receive another prompt and tools, and keep repeating this process with new prompts.
We run each experiment long enough for every agent to complete 32 such tasks, and record the performance metrics reported in Section~\ref{sec:metrics} over the full run.
During our experiments we use the gpt-oss-120b~\cite{openai_gpt-oss-120b_2025} model as our primary agent as it is a state-of-the-art reasoning model that might be commonly used in an agent.

We systematically sweep the main experimental hyperparameters to study their impact.
For tool behavior, we sample tool latencies from random normal distributions of two types:
\emph{short-latency} distributions with mean latencies $\in 
\{0, 0.1,\ldots,0.5\}$,  and \emph{long-latency} distributions with means $\in \{0,0.5,\ldots,3\}$.
For speculation, we vary the number of speculative samples per turn ($1, 3, 5, 7,$ and $9$) and the speculation model (xLAM 1B, 3B, and 8B).
For each point in this search space, we run the workload under three inference configurations: standard inference (no speculation), client-side speculation, and engine-side speculation.
Each of these experiments is run five times to account for variation and system noise.
Additionally, we run a subset of experiments with the client-side algorithm and commercial models gpt-5 and gpt-5-nano to study how well this approach works for commercial models.

\section{Evaluation Metrics}\label{sec:metrics}
With our tool-calling agent experiments set up, it is important that we are able to measure the benefits of speculation in order to compare approaches.
Metrics like walltime are tricky since some agents might decide to generate lots of response text to a prompt, while another gives a short response. This can happen due to diversity in sampled outputs, even for the same prompt.
For example, consider agents A and B each call tools, receive outputs from those tools, and then generate a response. 
Agent A has its tool speculated correctly, and agent B does not.
If agent A generates a 2000 token response to the tool output, while agent B generations a 50 token response, then agent A will have a larger wall time despite getting better benefits from speculation.
For this reason we use \emph{throughput} and \emph{percent time saved} as metrics for comparison.

\vspace{0.7em}
\noindent\textbf{Throughput} is computed as the number of tokens generated per second. We compute this for each agent and report the average throughput.

\vspace{0.7em}
\noindent\textbf{Time Saved} is a measure of percent reduced time relative to a non-speculative baseline. We record time from agent start to finish with and without speculation for the same prompt(s). Let $T_{\text{base}}$ denote this time for the baseline agent (no speculation) and $T_{\text{spec}}$ the time for the same agent configuration using our speculative algorithm. We define
\[
\text{Time Saved} = 100 \times \frac{T_{\text{base}} - T_{\text{spec}}}{T_{\text{base}}} \,\%.
\]
We compute this quantity per agent and report the mean across all the async agents. Because both runs share the same prompts, tools, and tool outputs, this metric isolates the reduction in inference time attributable to speculation.

In addition to throughput and time saved, we also analyze cost for the client side algorithm.
Since it can be implemented entirely using an API it can be used with commercial models.
We measure the \textbf{cost} using the token usage data returned by the OpenAI API and the publicly listed per-token prices from OpenAI (costs based on November 2025 pricing).
Cost is reported as the \emph{additional cost} of speculation versus just using the main model.
Since agents can vary in number of turns or tokens, we report the cost per 100 agent turns.

\section{Results}\label{sec:results}
In this section we present the results from our evaluation of the two proposed algorithms.

\subsection{Client-side Speculation}
\Cref{fig:client-side-acceptance-rates} shows the throughput of various speculative models with different acceptance rates for the client-side algorithm {\it on a subset of the BFCL prompts}. 
We show these results on an easier to predict subset of tools, so that we can observe the behavior for values $\alpha > 0.8$ to validate \Cref{eq:client-alg-speedup}. In the later experiments we observed $\alpha \approx 0.8$ for xLAM-2-8B.
Results are shown for runs with an average tool latency of 1.5 seconds, but the trends are similar for other tool latency amounts.
Point annotations denote how many tool calls are speculated by the speculative model each turn ($\lambda$ in Algorithm~\ref{alg:spec-tool-cache})(e.g. $1\times$, $3\times$, ...).

\begin{figure}[h]
    \centering
    \includegraphics[width=\linewidth]{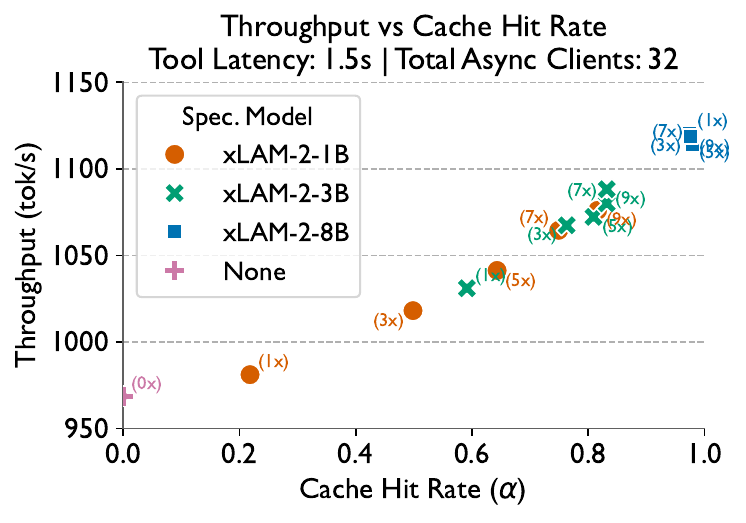}
    \caption{Throughput improvement for various speculative models and their cache hit respective tool cache hit rates. Points are annotated with the number of speculations per generation. We see a clear trend in improved throughput as the tool prediction rate increases. \label{fig:client-side-acceptance-rates}}
\end{figure}

We see that the trends from our performance model hold (see \Cref{sec:client-side-analysis}) hold; as speculation rate increases, we see an increase in inference throughput.
The 8B model yields much better accuracy in tool prediction than the smaller models, however, we see that the accuracy of the smaller models can also be improved by increasing the number of samples.
xLAM-2-1B achieves nearly 60\% higher hit rate with $9\times$ samples versus $1\times$ samples.
This flexibility is important as we may desire to run smaller speculative models due to memory constraints; taking more samples with a smaller model can get close to the performance of a larger speculative model.

\begin{figure}[ht]
    \centering
    \includegraphics[width=0.92\linewidth]{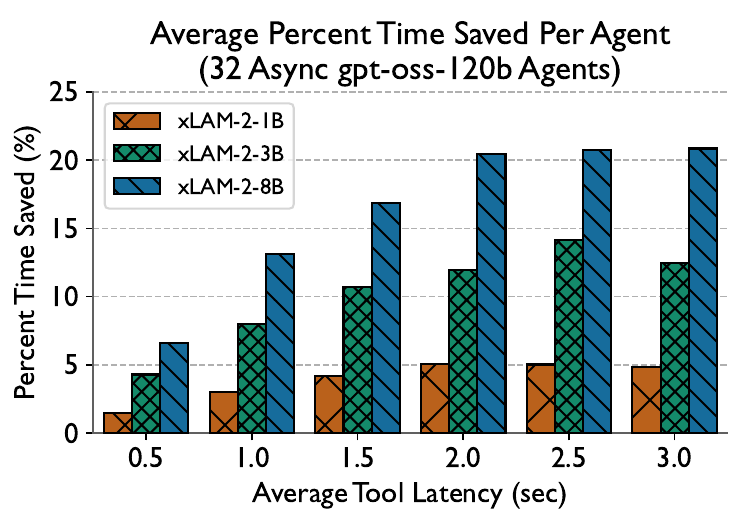}
    \caption{Average percent time saved across runs with various speculative models and average tool latencies. Results are shown for $1\times$ speculation per generation. We see that client side speculation can lead to between 6-21\% end-to-end time saving per agent when average tool latencies are 0.5 to 3 seconds. \label{fig:client-side-time-saved}}
\end{figure}

\Cref{fig:client-side-time-saved} shows the average percent time saved for each agent being run using the client side speculative algorithm. 
Even for shorter tools, around 0.5 seconds in duration, we still observe noticeable improvements with up to 6\% of total time saved by speculating the tools.
Results improve as the average tool time approaches gpt-oss-120b's average generation time; when tools average around 2.0-2.5 seconds we see the biggest gains with up to 21\% saved. 
The percent time saved slowly decrease as tools dominate the runtime and speculation only provides minor savings for long lasting tools.
This behavior is further evidence to the performance models in \Cref{sec:client-side-analysis} where we observe that tool calls around the generation time of the main model lead to the highest speedups.

\begin{observationbox}{Client-side Speculation Benefits}
An agentic framework, without any additional changes to the inference engine or API, can achieve 6-21\% time saved by speculating tools ahead of time. Savings vary depending on the ratio of model generation time and average tool call duration.
\end{observationbox}

\begin{figure}[h]
    \centering
    \includegraphics[width=\linewidth]{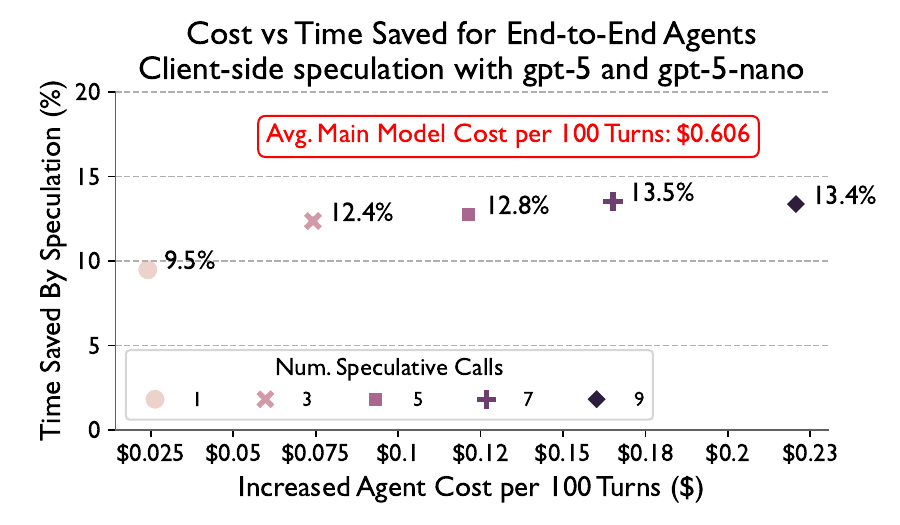}
    \caption{The increased cost for running client-side speculation with gpt-5-nano alongside a gpt-5 agent. Costs are shown for 100 agent turns. For only a 4\% price increase, agents can save nearly 10\% time.\label{fig:spec-cost}}
\end{figure}

We can also run the client-side speculation algorithm for commercial models where we only have API access.
\Cref{fig:spec-cost} shows the percent time saved by using speculation versus the added cost of doing that speculation.
The results are presented for a single agent ``turn", or one API then tool call pair.
We see that with one speculation we can get nearly 10\% time savings with only an extra API call that is an order of magnitude cheaper than the main API call.
Running nine speculations increases the time saved marginally, but is around 25-30\% of the cost of the main generation.

\subsection{Engine-side Speculation}

\Cref{fig:result-engine-side} shows the percent time saved for the engine-side algorithm.
We additionally show the times for average tool latencies in $\{0.1, 0.2, 0.3, 0.4, 0.5\}$ to highlight the regime where engine-side speculation helps most.
vLLM's speculative decoding infrastructure, what we build our implementation on top of, is currently in an experimental development phase~\cite{noauthor_speculative_nodate_v2} and we found high overheads for its spec-dec implementation when the batch size is greater than 1. 
For this reason we present our engine-side results with only one asynchronous agent.
Note that when there is only one asynchronous agent we see much lower generation latency from vLLM due to the request scheduler and smaller batch. Our observations align with findings in other works~\cite{singhania_llm_2025}. This leads to less potential gains as there is less generation time to overlap tool calls with, so we see lower percent time saved than in \Cref{fig:client-side-time-saved}.

\begin{figure}[h]
    \centering
    \includegraphics[width=\linewidth]{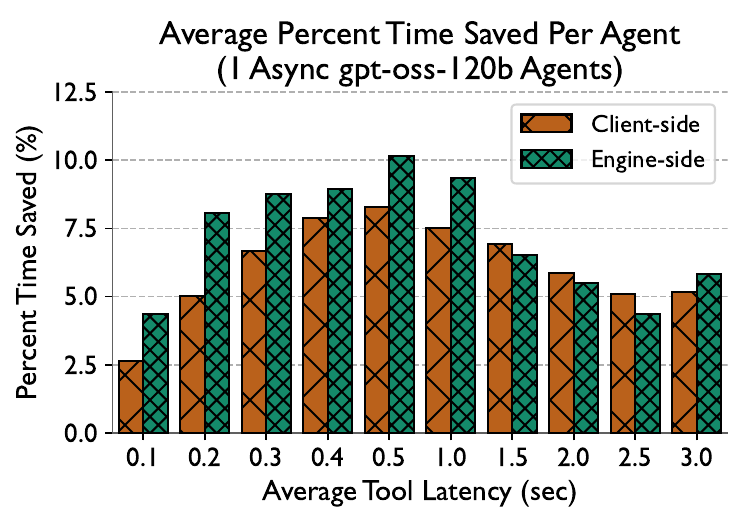}
    \caption{Average percent time saved comparison between client-side and engine-side algorithms for various average tool latencies. Results shown for one async client and xLAM-2-8B for speculation. When tools finish before the main model is done reasoning, we can achieve 2-3\% better time savings than the client-side approach.\label{fig:result-engine-side}}
\end{figure}

\begin{observationbox}{Engine-side Speculation Benefits}
Posting speculated tool outputs to an inference server to keep sequences resident when tool outputs are available a priori can lead to an additional 2-3\% time saved from the client's perspective.
The best benefits occur where tools latencies are less than the main model $\mathcal{M}$'s reasoning time.
This is ideal as many common agent tools fall in the 0-1 second range (e.g. \texttt{ls}, web search, read file, ...).
\end{observationbox}

We see 2-3\% increased time savings in the engine-side algorithm compared to the client-side algorithm. This is on top of the already high savings from the client-side speculation.
We see the best results where tool calls are between 0 and 1 second as these finish and are posted to the inference server before the reasoning phase is done.
Longer tool calls do not see any benefit over client-side speculation, since their results are not posted to the inference server before decoding is ready.

\section{Related Work}\label{sec:related-work}
Several works have looked at accelerating LM inference using speculation.
Leviathan et al.~\cite{leviathan_fast_2023} introduce speculative decoding to use a smaller LM to predict tokens for a larger LM.
While very related to our approach, the traditional speculative decoding does not work well for speculating tool calls.
As we have demonstrated in our analysis the main speedup from speculating tools comes from overlapping tool execution with generation.
In traditional speculative decoding the draft model is usually only ever 10 tokens ahead of the main model meaning we could only overlap tool calls with 10 decoding steps.
In the context of tool calling it is better to speculate tools as early and as many times as possible.
Other works~\cite{chen2023accelerating} have looked into other variations of generation speculation, such as in sampling, but suffer from the same drawback as speculative decoding: tool calls need to be speculated as soon as possible in the generation.

Other approaches to improving agent tool use performance try to increase the accuracy of smaller models at predicting the right tool to use. By using smaller LMs for taking actions we can reduce inference time. Schick et al.~\cite{schick2023toolformer} and Patil et al.~\cite{patil_gorilla_2023} accomplish this through fine-tuning on large datasets of tool calling interactions.
Works like SWE-Agent~\cite{yang2024swe} accomplish better and faster agent generation by providing the {\it best granularity} of tools to the agent. By giving it the right tools they reduce the number of tool calls it needs to make and ultimately the overall time for the agent to execute.

\section{Conclusion}
In this paper we showed that speculative execution of tools is an effective way to reduce the inference overheads of agents that use tools. We introduced two complementary techniques: a client-side speculative tool-calling algorithm that works with unmodified black-box APIs, and an engine-side algorithm that keeps prompts resident in the inference server by validating tool calls and ingesting tool outputs directly from a tool cache. With an analytical model we present the conditions where tool calling helps the most; we highlight the roles of tool latency, speculative model accuracy, and decode cost in determining how much speedup is achievable in practice. Our experiments show that these methods yield double-digit percent time savings per agent turn and increase throughput by hundreds of tokens per second. 
Speculative tools have the potential to rapidly accelerate the use of agents in tool heavy tasks, however, properly implementing speculative tool algorithms requires care to accomplish efficiently; this work provides the insights and building blocks to create effective speculative tool calling agents.

\section*{Acknowledgments}
This work was performed under the auspices of the U.S. Department of Energy (DOE) by Lawrence Livermore National Laboratory
under Contract DE-AC52-07NA27344 (LLNL-CONF-2014336-DRAFT). This
material is based upon work supported by the DOE Office of Science,
Advanced Scientific Computing Research program through solicitation DE-FOA-0003264, "Advancements in Artificial Intelligence
for Science."
This research used resources of the National Energy Research Scientific Computing Center (NERSC), a U.S. Department of Energy Office of Science User Facility, operated
under Contract No. DE-AC02-05CH11231 using NERSC
award ALCC-ERCAP0034775.

\bibliographystyle{plain}
\bibliography{citations/references,citations/extra}

@article{qu2025tool,
  title={Tool learning with large language models: A survey},
  author={Qu, Changle and Dai, Sunhao and Wei, Xiaochi and Cai, Hengyi and Wang, Shuaiqiang and Yin, Dawei and Xu, Jun and Wen, Ji-Rong},
  journal={Frontiers of Computer Science},
  volume={19},
  number={8},
  pages={198343},
  year={2025},
  publisher={Springer}
}

@article{qin2024tool,
author = {Qin, Yujia and Hu, Shengding and Lin, Yankai and Chen, Weize and Ding, Ning and Cui, Ganqu and Zeng, Zheni and Zhou, Xuanhe and Huang, Yufei and Xiao, Chaojun and Han, Chi and Fung, Yi Ren and Su, Yusheng and Wang, Huadong and Qian, Cheng and Tian, Runchu and Zhu, Kunlun and Liang, Shihao and Shen, Xingyu and Xu, Bokai and Zhang, Zhen and Ye, Yining and Li, Bowen and Tang, Ziwei and Yi, Jing and Zhu, Yuzhang and Dai, Zhenning and Yan, Lan and Cong, Xin and Lu, Yaxi and Zhao, Weilin and Huang, Yuxiang and Yan, Junxi and Han, Xu and Sun, Xian and Li, Dahai and Phang, Jason and Yang, Cheng and Wu, Tongshuang and Ji, Heng and Li, Guoliang and Liu, Zhiyuan and Sun, Maosong},
title = {Tool Learning with Foundation Models},
year = {2024},
issue_date = {April 2025},
publisher = {Association for Computing Machinery},
address = {New York, NY, USA},
volume = {57},
number = {4},
issn = {0360-0300},
url = {https://doi.org/10.1145/3704435},
doi = {10.1145/3704435},
journal = {ACM Comput. Surv.},
month = dec,
articleno = {101},
numpages = {40}
}

@article{liu2024large,
  title={Large language model-based agents for software engineering: A survey},
  author={Liu, Junwei and Wang, Kaixin and Chen, Yixuan and Peng, Xin and Chen, Zhenpeng and Zhang, Lingming and Lou, Yiling},
  journal={arXiv preprint arXiv:2409.02977},
  year={2024}
}

@article{yang2024swe,
  title={Swe-agent: Agent-computer interfaces enable automated software engineering},
  author={Yang, John and Jimenez, Carlos E and Wettig, Alexander and Lieret, Kilian and Yao, Shunyu and Narasimhan, Karthik and Press, Ofir},
  journal={Advances in Neural Information Processing Systems},
  volume={37},
  pages={50528--50652},
  year={2024}
}

@inproceedings{gao2024cost,
  title={$\{$Cost-Efficient$\}$ large language model serving for multi-turn conversations with $\{$CachedAttention$\}$},
  author={Gao, Bin and He, Zhuomin and Sharma, Puru and Kang, Qingxuan and Jevdjic, Djordje and Deng, Junbo and Yang, Xingkun and Yu, Zhou and Zuo, Pengfei},
  booktitle={2024 USENIX Annual Technical Conference (USENIX ATC 24)},
  pages={111--126},
  year={2024}
}

@inproceedings{kwon2023efficient,
  title={Efficient Memory Management for Large Language Model Serving with PagedAttention},
  author={Woosuk Kwon and Zhuohan Li and Siyuan Zhuang and Ying Sheng and Lianmin Zheng and Cody Hao Yu and Joseph E. Gonzalez and Hao Zhang and Ion Stoica},
  booktitle={Proceedings of the ACM SIGOPS 29th Symposium on Operating Systems Principles},
  year={2023}
}

@article{zheng2024sglang,
  title={Sglang: Efficient execution of structured language model programs},
  author={Zheng, Lianmin and Yin, Liangsheng and Xie, Zhiqiang and Sun, Chuyue Livia and Huang, Jeff and Yu, Cody Hao and Cao, Shiyi and Kozyrakis, Christos and Stoica, Ion and Gonzalez, Joseph E and others},
  journal={Advances in neural information processing systems},
  volume={37},
  pages={62557--62583},
  year={2024}
}

@article{chen2023accelerating,
  title={Accelerating large language model decoding with speculative sampling},
  author={Chen, Charlie and Borgeaud, Sebastian and Irving, Geoffrey and Lespiau, Jean-Baptiste and Sifre, Laurent and Jumper, John},
  journal={arXiv preprint arXiv:2302.01318},
  year={2023}
}

@inproceedings{yao2022react,
  title={React: Synergizing reasoning and acting in language models},
  author={Yao, Shunyu and Zhao, Jeffrey and Yu, Dian and Du, Nan and Shafran, Izhak and Narasimhan, Karthik R and Cao, Yuan},
  booktitle={The eleventh international conference on learning representations},
  year={2022}
}

@article{schick2023toolformer,
  title={Toolformer: Language models can teach themselves to use tools},
  author={Schick, Timo and Dwivedi-Yu, Jane and Dess{\`\i}, Roberto and Raileanu, Roberta and Lomeli, Maria and Hambro, Eric and Zettlemoyer, Luke and Cancedda, Nicola and Scialom, Thomas},
  journal={Advances in Neural Information Processing Systems},
  volume={36},
  pages={68539--68551},
  year={2023}
}

@article{chen2024facilitating,
  title={Facilitating multi-turn function calling for llms via compositional instruction tuning},
  author={Chen, Mingyang and Sun, Haoze and Li, Tianpeng and Yang, Fan and Liang, Hao and Lu, Keer and Cui, Bin and Zhang, Wentao and Zhou, Zenan and Chen, Weipeng},
  journal={arXiv preprint arXiv:2410.12952},
  year={2024}
}

@article{singh2024llm,
  title={An llm-tool compiler for fused parallel function calling},
  author={Singh, Simranjit and Karatzas, Andreas and Fore, Michael and Anagnostopoulos, Iraklis and Stamoulis, Dimitrios},
  journal={arXiv preprint arXiv:2405.17438},
  year={2024}
}

@article{song2025early,
  title={The early bird catches the leak: Unveiling timing side channels in llm serving systems},
  author={Song, Linke and Pang, Zixuan and Wang, Wenhao and Wang, Zihao and Wang, XiaoFeng and Chen, Hongbo and Song, Wei and Jin, Yier and Meng, Dan and Hou, Rui},
  journal={IEEE Transactions on Information Forensics and Security},
  year={2025},
  publisher={IEEE}
}

@misc{vllm_apc,
  author       = {vLLM},
  title        = {Automatic Prefix Caching},
  year         = {},
  publisher    = {},
  journal      = {},
  howpublished = {\url{https://docs.vllm.ai/en/stable/design/prefix_caching/}},
  commit       = {}
}

@misc{noauthor_speculative_nodate_v2,
	title = {Speculative {Decoding} - {vLLM}},
	url = {https://docs.vllm.ai/en/v0.12.0/features/spec_decode/},
	urldate = {2025-12-11},
    note = {\url{https://docs.vllm.ai/en/v0.12.0/features/spec_decode/} accessed on 2025-12-11},
}

@misc{openai_gpt-oss-120b_2025,
	title = {gpt-oss-120b \& gpt-oss-20b {Model} {Card}},
	url = {http://arxiv.org/abs/2508.10925},
	doi = {10.48550/arXiv.2508.10925},
	urldate = {2025-12-11},
	publisher = {arXiv},
	author = {OpenAI and Agarwal, Sandhini and Ahmad, Lama and Ai, Jason and Altman, Sam and Applebaum, Andy and Arbus, Edwin and Arora, Rahul K. and Bai, Yu and Baker, Bowen and Bao, Haiming and Barak, Boaz and Bennett, Ally and Bertao, Tyler and Brett, Nivedita and Brevdo, Eugene and Brockman, Greg and Bubeck, Sebastien and Chang, Che and Chen, Kai and Chen, Mark and Cheung, Enoch and Clark, Aidan and Cook, Dan and Dukhan, Marat and Dvorak, Casey and Fives, Kevin and Fomenko, Vlad and Garipov, Timur and Georgiev, Kristian and Glaese, Mia and Gogineni, Tarun and Goucher, Adam and Gross, Lukas and Guzman, Katia Gil and Hallman, John and Hehir, Jackie and Heidecke, Johannes and Helyar, Alec and Hu, Haitang and Huet, Romain and Huh, Jacob and Jain, Saachi and Johnson, Zach and Koch, Chris and Kofman, Irina and Kundel, Dominik and Kwon, Jason and Kyrylov, Volodymyr and Le, Elaine Ya and Leclerc, Guillaume and Lennon, James Park and Lessans, Scott and Lezcano-Casado, Mario and Li, Yuanzhi and Li, Zhuohan and Lin, Ji and Liss, Jordan and Lily and Liu and Liu, Jiancheng and Lu, Kevin and Lu, Chris and Martinovic, Zoran and McCallum, Lindsay and McGrath, Josh and McKinney, Scott and McLaughlin, Aidan and Mei, Song and Mostovoy, Steve and Mu, Tong and Myles, Gideon and Neitz, Alexander and Nichol, Alex and Pachocki, Jakub and Paino, Alex and Palmie, Dana and Pantuliano, Ashley and Parascandolo, Giambattista and Park, Jongsoo and Pathak, Leher and Paz, Carolina and Peran, Ludovic and Pimenov, Dmitry and Pokrass, Michelle and Proehl, Elizabeth and Qiu, Huida and Raila, Gaby and Raso, Filippo and Ren, Hongyu and Richardson, Kimmy and Robinson, David and Rotsted, Bob and Salman, Hadi and Sanjeev, Suvansh and Schwarzer, Max and Sculley, D. and Sikchi, Harshit and Simon, Kendal and Singhal, Karan and Song, Yang and Stuckey, Dane and Sun, Zhiqing and Tillet, Philippe and Toizer, Sam and Tsimpourlas, Foivos and Vyas, Nikhil and Wallace, Eric and Wang, Xin and Wang, Miles and Watkins, Olivia and Weil, Kevin and Wendling, Amy and Whinnery, Kevin and Whitney, Cedric and Wong, Hannah and Yang, Lin and Yang, Yu and Yasunaga, Michihiro and Ying, Kristen and Zaremba, Wojciech and Zhan, Wenting and Zhang, Cyril and Zhang, Brian and Zhang, Eddie and Zhao, Shengjia},
	month = aug,
	year = {2025},
	note = {arXiv:2508.10925 [cs]},
}

@misc{singhania_llm_2025,
	title = {{LLM} {Inference} {Beyond} a {Single} {Node}: {From} {Bottlenecks} to {Mitigations} with {Fast} {All}-{Reduce} {Communication}},
	shorttitle = {{LLM} {Inference} {Beyond} a {Single} {Node}},
	url = {http://arxiv.org/abs/2511.09557},
	doi = {10.48550/arXiv.2511.09557},
	urldate = {2025-12-11},
	publisher = {arXiv},
	author = {Singhania, Prajwal and Singh, Siddharth and Hough, Lannie Dalton and Srivastava, Akarsh and Menon, Harshitha and Jekel, Charles Fredrick and Bhatele, Abhinav},
	month = nov,
	year = {2025},
	note = {arXiv:2511.09557 [cs]},
}

@misc{zhang_xlam_2024,
	title = {{xLAM}: {A} {Family} of {Large} {Action} {Models} to {Empower} {AI} {Agent} {Systems}},
	shorttitle = {{xLAM}},
	url = {http://arxiv.org/abs/2409.03215},
	doi = {10.48550/arXiv.2409.03215},
	urldate = {2025-12-11},
	publisher = {arXiv},
	author = {Zhang, Jianguo and Lan, Tian and Zhu, Ming and Liu, Zuxin and Hoang, Thai and Kokane, Shirley and Yao, Weiran and Tan, Juntao and Prabhakar, Akshara and Chen, Haolin and Liu, Zhiwei and Feng, Yihao and Awalgaonkar, Tulika and Murthy, Rithesh and Hu, Eric and Chen, Zeyuan and Xu, Ran and Niebles, Juan Carlos and Heinecke, Shelby and Wang, Huan and Savarese, Silvio and Xiong, Caiming},
	month = sep,
	year = {2024},
	note = {arXiv:2409.03215 [cs]},
}

@misc{patil_gorilla_2023,
	title = {Gorilla: {Large} {Language} {Model} {Connected} with {Massive} {APIs}},
	shorttitle = {Gorilla},
	url = {http://arxiv.org/abs/2305.15334},
	doi = {10.48550/arXiv.2305.15334},
	urldate = {2025-12-11},
	publisher = {arXiv},
	author = {Patil, Shishir G. and Zhang, Tianjun and Wang, Xin and Gonzalez, Joseph E.},
	month = may,
	year = {2023},
	note = {arXiv:2305.15334 [cs]},
}

@misc{noauthor_openaiopenai-python_2025,
	title = {openai/openai-python},
	copyright = {Apache-2.0},
	url = {https://github.com/openai/openai-python},
	urldate = {2025-12-10},
	publisher = {OpenAI},
	month = dec,
	year = {2025},
	note = {original-date: 2020-10-25T23:23:54Z},
}

@misc{hu_speculative_2025,
	title = {Speculative {Decoding} and {Beyond}: {An} {In}-{Depth} {Survey} of {Techniques}},
	shorttitle = {Speculative {Decoding} and {Beyond}},
	url = {http://arxiv.org/abs/2502.19732},
	doi = {10.48550/arXiv.2502.19732},
	urldate = {2025-11-29},
	publisher = {arXiv},
	author = {Hu, Yunhai and Liu, Zining and Dong, Zhenyuan and Peng, Tianfan and McDanel, Bradley and Zhang, Sai Qian},
	month = oct,
	year = {2025},
	note = {arXiv:2502.19732 [cs]},
}

@inproceedings{leviathan_fast_2023,
	address = {Honolulu, Hawaii, USA},
	series = {{ICML}'23},
	title = {Fast inference from transformers via speculative decoding},
	volume = {202},
	urldate = {2025-11-29},
	booktitle = {Proceedings of the 40th {International} {Conference} on {Machine} {Learning}},
	publisher = {JMLR.org},
	author = {Leviathan, Yaniv and Kalman, Matan and Matias, Yossi},
	month = jul,
	year = {2023},
	pages = {19274--19286},
}

\end{document}